\documentclass[%
 twocolumn,
 aps,pra,
 amsmath,amssymb,
 superscriptaddress,
 nofootinbib,
 reprint,%
]{revtex4-1}
\usepackage{stmaryrd}
\usepackage{color}
\usepackage{array}
\usepackage{enumitem}
\usepackage{mathpazo}
\usepackage{setspace}
\usepackage{bm}
\usepackage{amssymb}
\usepackage{amsthm}
\usepackage{mathtools}
\usepackage{multirow}
\usepackage{cases}
\usepackage[colorlinks=true,linkcolor=blue,citecolor=blue,pdfauthor={ },pdftitle={ },pdfsubject={ },pdfkeywords={ }]{hyperref}
\usepackage{pgfplots}
\usepackage{lipsum}
\usepackage{youngtab}

\newtheorem{definition}{Definition}
\newtheorem{proposition}[definition]{Proposition}
\newtheorem{lemma}[definition]{Lemma}

\newtheorem{fact}[definition]{Fact}
\newtheorem{theorem}[definition]{Theorem}
\newtheorem{observation}[definition]{Observation}
\newtheorem{corollary}[definition]{Corollary}
\newtheorem{conjecture}[definition]{Conjecture}

\newtheorem{remark}[definition]{Remark}
\newtheorem{example}[definition]{Example}
\newtheorem{question}[definition]{Question}

\def\Dbar{\leavevmode\lower.6ex\hbox to 0pt
{\hskip-.23ex\accent"16\hss}D}
\makeatletter
\def\url@leostyle{%
  \@ifundefined{selectfont}{\def\UrlFont{\sf}}{\def\UrlFont{\small\ttfamily}}}
\makeatother
\DeclareMathOperator{\Tr}{Tr} %

\def\bcj{\begin{conjecture}}
\def\ecj{\end{conjecture}}
\def\bcr{\begin{corollary}}
\def\ecr{\end{corollary}}
\def\bd{\begin{definition}}
\def\ed{\end{definition}}
\def\bea{\begin{eqnarray}}
\def\eea{\end{eqnarray}}
\def\bem{\begin{enumerate}}
\def\eem{\end{enumerate}}
\def\bex{\begin{example}}
\def\eex{\end{example}}
\def\bim{\begin{itemize}}
\def\eim{\end{itemize}}
\def\bl{\begin{lemma}}
\def\el{\end{lemma}}
\def\bpf{\begin{proof}}
\def\epf{\end{proof}}
\def\bpp{\begin{proposition}}
\def\epp{\end{proposition}}
\def\bqu{\begin{question}}
\def\equ{\end{question}}
\def\br{\begin{remark}}
\def\er{\end{remark}}
\def\bt{\begin{theorem}}
\def\et{\end{theorem}}

\def\btb{\begin{tabular}}
\def\etb{\end{tabular}}

\newcommand{\nc}{\newcommand}

 \nc{\bA}{{\bf A}} \nc{\bB}{{\bf B}} \nc{\bC}{{\bf C}}
 \nc{\bD}{{\bf D}} \nc{\bE}{{\bf E}} \nc{\bF}{{\bf F}}
 \nc{\bG}{{\bf G}} \nc{\bH}{{\bf H}} \nc{\bI}{{\bf I}}
 \nc{\bJ}{{\bf J}} \nc{\bK}{{\bf K}} \nc{\bL}{{\bf L}}
 \nc{\bM}{{\bf M}} \nc{\bN}{{\bf N}} \nc{\bO}{{\bf O}}
 \nc{\bP}{{\bf P}} \nc{\bQ}{{\bf Q}} \nc{\bR}{{\bf R}}
 \nc{\bS}{{\bf S}} \nc{\bT}{{\bf T}} \nc{\bU}{{\bf U}}
 \nc{\bV}{{\bf V}} \nc{\bW}{{\bf W}} \nc{\bX}{{\bf X}}
 \nc{\bZ}{{\bf Z}}

\nc{\cA}{{\cal A}} \nc{\cB}{{\cal B}} \nc{\cC}{{\cal C}}
\nc{\cD}{{\cal D}} \nc{\cE}{{\cal E}} \nc{\cF}{{\cal F}}
\nc{\cG}{{\cal G}} \nc{\cH}{{\cal H}} \nc{\cI}{{\cal I}}
\nc{\cJ}{{\cal J}} \nc{\cK}{{\cal K}} \nc{\cL}{{\cal L}}
\nc{\cM}{{\cal M}} \nc{\cN}{{\cal N}} \nc{\cO}{{\cal O}}
\nc{\cP}{{\cal P}} \nc{\cQ}{{\cal Q}} \nc{\cR}{{\cal R}}
\nc{\cS}{{\cal S}} \nc{\cT}{{\cal T}} \nc{\cU}{{\cal U}}
\nc{\cV}{{\cal V}} \nc{\cW}{{\cal W}} \nc{\cX}{{\cal X}}
\nc{\cZ}{{\cal Z}}

\nc{\hA}{{\hat{A}}} \nc{\hB}{{\hat{B}}} \nc{\hC}{{\hat{C}}}
\nc{\hD}{{\hat{D}}} \nc{\hE}{{\hat{E}}} \nc{\hF}{{\hat{F}}}
\nc{\hG}{{\hat{G}}} \nc{\hH}{{\hat{H}}} \nc{\hI}{{\hat{I}}}
\nc{\hJ}{{\hat{J}}} \nc{\hK}{{\hat{K}}} \nc{\hL}{{\hat{L}}}
\nc{\hM}{{\hat{M}}} \nc{\hN}{{\hat{N}}} \nc{\hO}{{\hat{O}}}
\nc{\hP}{{\hat{P}}} \nc{\hR}{{\hat{R}}} \nc{\hS}{{\hat{S}}}
\nc{\hT}{{\hat{T}}} \nc{\hU}{{\hat{U}}} \nc{\hV}{{\hat{V}}}
\nc{\hW}{{\hat{W}}} \nc{\hX}{{\hat{X}}} \nc{\hZ}{{\hat{Z}}}

\newcommand{\bra}[1]{\langle#1|}
\newcommand{\ket}[1]{|#1\rangle}

\def\Dbar{\leavevmode\lower.6ex\hbox to 0pt
{\hskip-.23ex\accent"16\hss}D}
\begin{document}


\def\be{\begin{eqnarray}}
\def\ee{\end{eqnarray}}


\newcommand{\ca}{\mathcal A}

\newcommand{\cb}{\mathcal B}
\newcommand{\cc}{\mathcal C}
\newcommand{\cd}{\mathcal D}
\newcommand{\ce}{\mathcal E}
\newcommand{\cf}{\mathcal F}
\newcommand{\cg}{\mathcal G}
\newcommand{\ch}{\mathcal H}
\newcommand{\ci}{\mathcal I}
\newcommand{\cj}{\mathcal J}
\newcommand{\ck}{\mathcal K}
\newcommand{\cl}{\mathcal L}
\newcommand{\cm}{\mathcal M}
\newcommand{\cn}{\mathcal N}
\newcommand{\co}{\mathcal O}
\newcommand{\cp}{\mathcal P}
\newcommand{\cq}{\mathcal Q}
\newcommand{\calr}{\mathcal R}
\newcommand{\cs}{\mathcal S}
\newcommand{\ct}{\mathcal T}
\newcommand{\cu}{\mathcal U}
\newcommand{\cv}{\mathcal V}
\newcommand{\cw}{\mathcal W}
\newcommand{\cx}{\mathcal X}
\newcommand{\cy}{\mathcal Y}
\newcommand{\cz}{\mathcal Z}


\newcommand{\sa}{\mathscr{A}}
\newcommand{\sm}{\mathscr{M}}


\newcommand{\fa}{\mathfrak{a}}  \newcommand{\Fa}{\mathfrak{A}}
\newcommand{\fb}{\mathfrak{b}}  \newcommand{\Fb}{\mathfrak{B}}
\newcommand{\fc}{\mathfrak{c}}  \newcommand{\Fc}{\mathfrak{C}}
\newcommand{\fd}{\mathfrak{d}}  \newcommand{\Fd}{\mathfrak{D}}
\newcommand{\fe}{\mathfrak{e}}  \newcommand{\Fe}{\mathfrak{E}}
\newcommand{\ff}{\mathfrak{f}}  \newcommand{\Ff}{\mathfrak{F}}
\newcommand{\fg}{\mathfrak{g}}  \newcommand{\Fg}{\mathfrak{G}}
\newcommand{\fh}{\mathfrak{h}}  \newcommand{\Fh}{\mathfrak{H}}
\newcommand{\fraki}{\mathfrak{i}}       \newcommand{\Fraki}{\mathfrak{I}}
\newcommand{\fj}{\mathfrak{j}}  \newcommand{\Fj}{\mathfrak{J}}
\newcommand{\fk}{\mathfrak{k}}  \newcommand{\Fk}{\mathfrak{K}}
\newcommand{\fl}{\mathfrak{l}}  \newcommand{\Fl}{\mathfrak{L}}
\newcommand{\fm}{\mathfrak{m}}  \newcommand{\Fm}{\mathfrak{M}}
\newcommand{\fn}{\mathfrak{n}}  \newcommand{\Fn}{\mathfrak{N}}
\newcommand{\fo}{\mathfrak{o}}  \newcommand{\Fo}{\mathfrak{O}}
\newcommand{\fp}{\mathfrak{p}}  \newcommand{\Fp}{\mathfrak{P}}
\newcommand{\fq}{\mathfrak{q}}  \newcommand{\Fq}{\mathfrak{Q}}
\newcommand{\fr}{\mathfrak{r}}  \newcommand{\Fr}{\mathfrak{R}}
\newcommand{\fs}{\mathfrak{s}}  \newcommand{\Fs}{\mathfrak{S}}
\newcommand{\ft}{\mathfrak{t}}  \newcommand{\Ft}{\mathfrak{T}}
\newcommand{\fu}{\mathfrak{u}}  \newcommand{\Fu}{\mathfrak{U}}
\newcommand{\fv}{\mathfrak{v}}  \newcommand{\Fv}{\mathfrak{V}}
\newcommand{\fw}{\mathfrak{w}}  \newcommand{\Fw}{\mathfrak{W}}
\newcommand{\fx}{\mathfrak{x}}  \newcommand{\Fx}{\mathfrak{X}}
\newcommand{\fy}{\mathfrak{y}}  \newcommand{\Fy}{\mathfrak{Y}}
\newcommand{\fz}{\mathfrak{z}}  \newcommand{\Fz}{\mathfrak{Z}}

\newcommand{\cfg}{\dot \fg}
\newcommand{\cFg}{\dot \Fg}
\newcommand{\ccg}{\dot \cg}
\newcommand{\circj}{\dot {\mathbf J}}
\newcommand{\circs}{\circledS}
\newcommand{\jmot}{\mathbf J^{-1}}


\newcommand{\rmd}{\mathrm d}
\newcommand{\mca}{\ ^-\!\!\ca}
\newcommand{\pca}{\ ^+\!\!\ca}
\newcommand{\peq}{^\Psi\!\!\!\!\!=}
\newcommand{\lt}{\left}
\newcommand{\rt}{\right}
\newcommand{\HN}{\hat{H}(N)}
\newcommand{\HM}{\hat{H}(M)}
\newcommand{\Hv}{\hat{H}_v}
\newcommand{\cyl}{\mathbf{Cyl}}
\newcommand{\lag}{\left\langle}
\newcommand{\rag}{\right\rangle}
\newcommand{\Ad}{\mathrm{Ad}}
\newcommand{\trace}{\mathrm{tr}}
\newcommand{\bbc}{\mathbb{C}}
\newcommand{\AC}{\overline{\mathcal{A}}^{\mathbb{C}}}
\newcommand{\Ar}{\mathbf{Ar}}
\newcommand{\uc}{\mathrm{U(1)}^3}
\newcommand{\M}{\hat{\mathbf{M}}}
\newcommand{\spin}{\text{Spin(4)}}
\newcommand{\id}{\mathrm{id}}
\newcommand{\Pol}{\mathrm{Pol}}
\newcommand{\Fun}{\mathrm{Fun}}
\newcommand{\bp}{p}
\newcommand{\act}{\rhd}
\newcommand{\data}{\lt(j_{ab},A,\bar{A},\xi_{ab},z_{ab}\rt)}
\newcommand{\datao}{\lt(j^{(0)}_{ab},A^{(0)},\bar{A}^{(0)},\xi_{ab}^{(0)},z_{ab}^{(0)}\rt)}
\newcommand{\deltadata}{\lt(j'_{ab}, A',\bar{A}',\xi_{ab}',z_{ab}'\rt)}
\newcommand{\background}{\lt(j_{ab}^{(0)},g_a^{(0)},\xi_{ab}^{(0)},z_{ab}^{(0)}\rt)}
\newcommand{\sgn}{\mathrm{sgn}}
\newcommand{\vth}{\vartheta}
\newcommand{\rmi}{\mathrm{i}}
\newcommand{\bfmu}{\pmb{\mu}}
\newcommand{\bfnu}{\pmb{\nu}}
\newcommand{\bfm}{\mathbf{m}}
\newcommand{\bfn}{\mathbf{n}}


\newcommand{\sz}{\mathscr{Z}}
\newcommand{\sk}{\mathscr{K}}

\title{Symmetric vs. bosonic extension for bipartite states}

\author{Youning Li}%
\affiliation{College of Science, China Agricultural University, Beijing, 100080, People's Republic of China}
\affiliation{Department of Physics, Tsinghua University, Beijing, People's Republic of China}%
\affiliation{Collaborative Innovation Center of Quantum Matter, Beijing 100190, People's Republic of China}%
\author{Shilin Huang}
\affiliation{Institute for Interdisciplinary Information Sciences, Tsinghua University, Beijing, People's Republic of China}
\affiliation{Department of Electrical and Computer Enginnering, Duke University, Durham, NC, 27708, USA}

\author{Dong Ruan}
\affiliation{Department of Physics, Tsinghua University, Beijing, People's Republic of China}%
\affiliation{Collaborative Innovation Center of Quantum Matter, Beijing 100190, People's Republic of China}%
\author{Bei Zeng}
\affiliation{Department of Mathematics \& Statistics, University of
  Guelph, Guelph, Ontario, Canada}%
\affiliation{Institute for Quantum Computing and Department of Physics and Astronomy, University of Waterloo, Waterloo, Ontario, Canada}

\date{\today}

\begin{abstract}
A bipartite state $\rho^{AB}$ has a $k$-symmetric extension if there exists a $k+1$-partite state $\rho^{AB_1B_2\ldots B_k}$ with marginals $\rho^{AB_i}=\rho^{AB}, \forall i$. The $k$-symmetric extension is called bosonic
if $\rho^{AB_1B_2\ldots B_k}$ is supported on the symmetric subspace of $B_1B_2\ldots B_k$. Understanding the structure of symmetric/bosonic extension has various applications in the theory of quantum entanglement, quantum key distribution and the quantum marginal problem. In particular, bosonic extension gives a tighter bound for the quantum marginal problem based on seperability. In general, it is known that a $\rho^{AB}$ admitting symmetric extension may not have bosonic extension. In this work, we show that when the dimension of the subsystem $B$ is $2$ (i.e. a qubit), $\rho^{AB}$ admits a $k$-symmetric extension if and only if it has a $k$-bosonic extension. Our result has an immediate application to the quantum marginal problem and indicates a special structure for qubit systems based on group representation theory.
\end{abstract}

\maketitle
\renewcommand\theequation{\arabic{section}.\arabic{equation}}
\setcounter{tocdepth}{4}
\makeatletter
\@addtoreset{equation}{section}
\makeatother

\section{Introduction}


Entanglement is one of the central mysteries of quantum mechanics --
two or more parties can be correlated in the way that is much stronger than they
can be in any classical way~\cite{Horodecki2009}. Due to its striking features, entanglement plays a key role  in many quantum information processing tasks such as teleportation and quantum key distribution~\cite{NC00}. However, while entanglement   has been investigated fairly extensively in the research literature, identifying entangled state remains a challenging task. Indeed, even for  bipartite  quantum systems, there is   no  generic procedure that can tell us whether a given bipartite state is entangled or not. Actually,  the entanglement detection problem has long been known to be NP-hard in general~\cite{gurvits2003classical}.

Consider a bipartite quantum system with Hilbert space
$\mathbb{C}^{d_A}\otimes\mathbb{C}^{d_B}$, here the subsystems are labeled $A$ and $B$.
A  state
$\rho^{AB}$ is separable if it can be written as the convex
combination $\sum_i p_i \rho^{A,i}\otimes \rho^{B,i}$ for a
probability distribution $p_i$ and states $\rho^{A,i}$ and
$\rho^{B,i}$, otherwise it is entangled~\cite{werner1989quantum}. In practice, one typically constructs detection criteria based on simple properties that are obeyed by all separable states, therefore these are necessary but not sufficient conditions for separability. A most favoured approach is  the known partial transpose (PPT) criterion~\cite{Per96,HHH96}.   Another commonly
used method is through the $k$-symmetric extension hierarchy.

A bipartite state $\rho^{AB}$ is $k$-symmetric
extendible if there is a global quantum state
$\rho^{AB_1B_2\cdots B_k}$ whose marginals on $A, B_i$ equal to the
given $\rho_{AB}$ for $i=1,2,\ldots, k$. It was found that, the set of all $k$-extendible
states, denoted by $\Theta_k$, is convex, with a hierarchy structure
$\Theta_k\supset\Theta_{k+1}$, and besides, in the $k\rightarrow\infty$ limit,
$\Theta_{\infty}$ converges exactly to the
set of separable states which is also convex~\cite{doherty02a}.   In other words, separable states are the only states that have $k$-copy symmetric extensions for all $k \ge 2$. This leads to a separability criteria which consists of a hierarchy of tests:   one asks about whether or not a given state belongs to the $k$-extendible set $\Theta_{k}$ for increasing $k$.

A bipartite state $\rho^{AB}$ is $k$-bosonic
extendible if the global quantum state
$\rho^{AB_1B_2\cdots B_k}$ with $\rho^{AB_i}=\rho^{AB}$
is supported on the symmetric subspace of $B_1B_2\ldots B_N$.
Similarly, the set of all all $k$-extendible
states, denoted by $\bar{\Theta}_k$, is convex, with a hierarchy structure
$\bar{\Theta}_k\supset\bar{\Theta}_{k+1}$, and in the $k\rightarrow\infty$ limit,
$\bar{\Theta}_{\infty}=\Theta_{\infty}$ converges also to the
set of separable states. Obviously $\bar{\Theta}_k\subseteq{\Theta}_{k}$, as a seperability
the $k$-bosonic extension is stronger than the $k$-symmetric extension.
Based on $k$-symmetric/bosonic extension, effective numerical tests for separability has been developed~\cite{navascues2009power,brandao2011quasipolynomial,johnston2014detection}.

It is natural to ask whether $\bar{\Theta}_k$ is strictly contained in ${\Theta}_{k}$ for any finite
$k$. It turns out that the answer depends on $k$ and the dimension of the system $B$. It is known that
$\bar{\Theta}_2={\Theta}_{2}$ for $d_A=d_B=2$, and $\bar{\Theta}_2\subset{\Theta}_{2}$ for $d_A=d_B=3$~\cite{myhr2011symmetric}. An example of $\rho^{AB}$ with $2$-symmetric extension that has no bosonic extension can be constructed from an pure $3$-qutrit state $\rho^{AB_1B_2}$ that is supported on the antisymmetric subspace of $B_1B_2$. This may indicate that $\bar{\Theta}_k\subset{\Theta}_{k}$ for $d_B>2$. In this sense, the $d_B=2$ case is of particular interest, given that $\bar{\Theta}_k={\Theta}_{k}$
for $k=2$ and $\infty$. One would naturally wonder whether it is also the case for any other $k$. Our main result of this work, as summarized below, shows it is indeed the case.

\textbf{Main result:} For $d_B=2$, $\rho^{AB}$ admits a $k$-symmetric extension if and only if
it has a $k$-bosonic extension, for any $k$. That is, $\bar{\Theta}_k={\Theta}_{k}$ for $d_B=2$.

This result finds an immediate application to the quantum marginal problem,
also known as the consistency problem, which
asks for the conditions under which there exists an $N$-particle
density matrix $\rho_N$ whose reduced density matrices (quantum
marginals) on the subsets of particles $S_i \subset \{1,2,\ldots, N\}$
equal to the given density matrices $\rho_{S_i}$ for all
$i$~\cite{Kly06}. The related problem in fermionic (bosonic) systems
is the so-called $N$-representability problem,
which inherits a long history in quantum
chemistry~\cite{Col63,Erd72}.

In this sense, the $k$-symmetric extension problem is a special
case of the quantum marginal problem, and the $k$-bosonic extension problem
is intimately related to the $N$-representability problem~\cite{chen2014symmetric}. And it worth mentioning
that the quantum marginal problem and the $N$-representability problem are
in general very difficult. They were shown to be the complete problems
of the complexity class QMA, even for the relatively simple case where
the given marginals are two-particle states~\cite{Liu06,LCV07,WMN10}.
In other words, even with the help of a quantum computer, it is very
unlikely that the quantum marginal problems can be solved efficiently
in the worst case.

An interesting necessary condition of the $k$-symmetric/bosonic extension
problem is derived in~\cite{chen2016detecting}, based on the
separability of $\rho^{AB}$. It shows that if $\rho^{AB}$
has $k$-symmetric extension then the state
$
\tilde{\rho}^{AB}_{k}=
\left(d_B\rho^{A}\otimes I_B+k\rho^{AB}\right)/(d_B^2+k)
$
is separable. This condition can be strengthened if $\rho^{AB}$
has $k$-bosonic extension, where the state
$
\tilde{\rho}^{AB}_{k}=
\left(\rho^{A}\otimes I_B+k\rho^{AB}\right)/(d_B+k)
$
is separable.

Our main result hence has an immediate corollary as summarized below.
And it has shown that this result leads to strong conditions for detecting the consistency
of overlapping marginals~\cite{chen2016detecting}.

\textbf{Corollary} For $d_B=2$, if $\rho^{AB}$
has $k$-symmetric extension, then the state
$
\tilde{\rho}^{AB}_{k}=
\left(\rho^{A}\otimes I_B+k\rho^{AB}\right)/(k+2)
$
is separable.

We organize our paper as follows: in Sec.~\ref{sec:2}, we review some background of the and known results for the relationship between $\bar{\Theta}_k$ and ${\Theta}_{k}$; in Sec.~\ref{sec:3}, we use the $d_B=2, k=3$ case as an example to demonstrate the proof idea of our main result; in Sec.~\ref{sec:4}, we discuss the proof idea in for the general case; some further discussions are given in Sec.~\ref{sec:5}; some technical details of the proof are discussed in the appendices.

\section{Background and previous results}
\label{sec:2}

Consider the following notations
\begin{subequations}
\begin{eqnarray}
\Tr_{B_2B_3\ldots B_k}[\rho^{AB_1\cdots B_k}]&=&\rho^{AB_1},\\
\left(\mathbb{1}^A\otimes P^{ij}\right)\rho^{AB_1\cdots B_k}\left(\mathbb{1}^A\otimes P^{ij}\right)^{\dagger}&=&\rho^{AB_1\cdots B_k},\label{symmetric}
\end{eqnarray}
\end{subequations}
where the operator $P^{ij} \in S_k$ is an element in permutation group $S_k$, which swaps the $i$th subsystem $B_i$ and the $j$th subsystem $B_j$.
The global state $\rho^{AB_1\cdots B_k}$ is called a $k$-symmetric extension of $\rho^{AB_1}$.

Eq.(\ref{symmetric}) requires the global state $\rho^{AB_1\cdots B_k}$ is invariant under any exchange of $B_i$ and $B_j$, but it does not require that $\rho^{AB_1\cdots B_k}$ must support on a subspace with specific permutation symmetry.
e.g. for a $2$-symmetric extendible state, its extension can be bosonic, which supports on the symmetric subspace only, or fermionic, whose support only resides on the antisymmetric subspace, or more generally, can be a mixture of both.

It has already been known that
\begin{fact}\label{2extionsion}
Given any $2$-symmetric extendible state $\rho^{AB}$, if $d_B=2$, then a bosonic extension always exists.\footnote{This claim does not always work when the subsystem has higher dimension.}
\end{fact}
The original proof can be found in~\cite{myhr2009symmetric}.
For consistence and readability, we include the proof here.
\begin{proof}
For $k=2$, Eq.(\ref{symmetric}) reduces to
\begin{equation}
\left(\mathbb{1}^A\otimes P^{B_1 B_2}\right)\rho^{A B_1 B_2}\left(\mathbb{1}^A\otimes P^{B_1 B_2}\right)^{\dagger}=\rho^{A B_1 B_2},
\end{equation}
which means that $\rho^{A B_1 B_2}$ commutes with $(\mathbb{1}_A\otimes P_{B_1 B_2})$.
Therefore, they have common eigen-vectors, say $\{\ket{\phi_j}\}$.
Since $\left(\mathbb{1}^A\otimes P^{B_1 B_2}\right)^2=\mathbb{1}$,  we have
\begin{eqnarray}\label{spectra}
\mathbb{1}^A\otimes P^{B_1 B_2}\ket{\phi_j}=\pm \ket{\phi_j},\,\forall j.
\end{eqnarray}
Thus generically, $\rho^{A B_1 B_2}$ can be decomposed as
\begin{equation}
\rho^{A B_1 B_2}=\sum_j \lambda_j^+ \ket{\phi_j^+}\bra{\phi_j^+} + \sum_l \lambda_l^- \ket{\phi_l^-}\bra{\phi_l^-},
\end{equation}
where $\mathbb{1}^A\otimes P^{B_1 B_2}\ket{\phi_j^{\pm}}=\pm \ket{\phi_j^{\pm}}$.
Owning to the fact that $B_1$ and $B_2$ are two qubits, $\ket{\phi_j^{\pm}}$ can be further decomposed as
\begin{subequations}
\begin{eqnarray}
\ket{\phi_j^+}&=&\sum_{k}\ket{\psi_{j,k}}_A\ket{\psi_k^+}_{B_1 B_2},\\
\ket{\phi_j^-}&=&\ket{\xi_{j}}_A\ket{\psi^-}_{B_1 B_2},
\end{eqnarray}
\end{subequations}
where $\ket{\psi_{j,k}}_A$ and $\ket{\xi_{j}}_A$ are vectors of subsystem $A$, while $\ket{\psi_k^+}_{B_1 B_2}$ and $\ket{\psi^-}_{B_1 B_2}$ are the triplet and singlet states respectively.\footnote{Here we do not require $\ket{\psi_{j,k}}_A$ and $\ket{\xi_{j}}_A$ to be normalized for simplicity in description.}
Replacing the singlet state $\ket{\psi^-}_{B_1 B_2}\equiv\frac{1}{\sqrt 2}\left(\ket{01}-\ket{10}\right)$ with $\ket{\psi^+}_{B_1 B_2}\equiv\frac{1}{\sqrt 2}\left(\ket{01}+\ket{10}\right)$, we get a new global state $\sigma^{A B_1 B_2}$
\begin{eqnarray}\label{decomposition}
\sigma^{A B_1 B_2}&=&
\sum_{j}
\lambda_j^+ \sum_{k,k'}\ket{\psi_{j,k}}_A\bra{\psi_{j',k'}}_A \ket{\psi_k^+}_{B_1 B_2} \bra{\psi_{k'}^+}_{B_1 B_2}\notag\\
 &+& \sum_l \lambda_l^- \ket{\xi_l^-}_A\bra{\xi_l^-}_A\ket{\psi^+}_{B_1 B_2}\bra{\psi^+}_{B_1 B_2}.
\end{eqnarray}
Obviously, $\sigma^{A B_1 B_2}$ and $\rho^{A B_1 B_2}$ have identical reduced density matrix $\rho^{A B_1}$, and $\sigma^{A B_1 B_2}$ supports on the bosonic subspace.
\end{proof}

The above proof can be roughly divided into $2$ steps:
\begin{enumerate}
\item Find the general form of global state after symmetric extension, which probably is a convex combination of bosonic extension and non-bosonic extension.
\item Demonstrate that a bosonic extension, which preserve the reduced density matrix untouched, could be yielded by replacing the non-bosonic component with a bosonic one.
\end{enumerate}

However, the above $2$-symmetric extension possesses properties that are not true for general $k$.
\begin{itemize}
\item[a.] When considering $2$-symmetric extension states, the permutation group contains only $1$ non-trivial elements $P^{B_1 B_2}$, thus all permutations commute with global density matrix and have the common set of eigen-vectors. While for general $k$, the permutation group itself is a non-abelian group, thus Eq.(\ref{spectra}) will not always hold.
\item[b.] The dimension of non-bosinic subspace in $2$-qubit is $1$, hence we do not have to consider the off diagonal terms for non-bosonic component. Again, it is no longer true for general $k$.
\end{itemize}
The above differences implies that the decomposition of general $k$-symmetric extendible states after extension will be more complicated than Eq.(\ref{decomposition}).

\section{$3$-symmetric/bosonic extension}
\label{sec:3}
Before starting the proof for general $k$, we first take a look at the $k=3$ case.

Consider the Hilbert space $\ct\equiv V^{(1)}\otimes V^{(2)}\otimes V^{(3)}$ constituted by $3$-qubit $B_1, B_2$ and $B_3$, where each $V^{(i)}$ represents a qubit. $\ct$ is spanned by $8$ vectors: $\left\{\ket{000}, \ket{001},\cdots, \ket{111}\right\}$.
The bosonic subspace contains $4$ linear independent vectors
\begin{eqnarray*}
&&\qquad\qquad\ket{000},\\
&&\frac{1}{\sqrt{3}}\left(\ket{001}+\ket{010}+\ket{100}\right),\\
&&\frac{1}{\sqrt{3}}\left(\ket{011}+\ket{101}+\ket{110}\right),\\
&&\qquad\qquad\ket{111}.
\end{eqnarray*}
Since the cross term of bosonic subspace and non-bosonic subspace is forbidden, we shall only consider the density matrix that supports on nonbosonic subspace.

Notice that, permutation can only swap two or more subsystems, but keep the numbers of $\ket{0}$ and $\ket{1}$ constant. Thus the nonbosonic subspace could be further divided into $2$ subspaces $V^{(a)}$ and $V^{(b)}$.
$V^{(a)}$ has $2$ $\ket{0}$ and $1$ $\ket{1}$ while the other has $1$ $\ket{0}$ and $2$ $\ket{1}$.

First let us consider a density matrix $\tilde{\rho}^{AB_1B_2B_3}$ that purely supports on $\mathrm{End}(V_A)\otimes \mathrm{End}(V^{(a)})$.
The dimension of this subspace is $2$ and one can find a basis
\begin{subequations}
\begin{eqnarray}
\ket{\psi^{(a)}_1} &\equiv& \frac{1}{\sqrt{6}}\left(2\ket{001}-\ket{010}-\ket{100}\right)\\
\ket{\psi^{(a)}_2} &\equiv& \frac{1}{\sqrt{2}}\left(\ket{010}-\ket{100}\right)
\end{eqnarray}
\end{subequations}
It can be easily checked that, if Eq.(\ref{symmetric}) was satisfied,
\begin{eqnarray}
\tilde{\rho}^{AB_1B_2B_3}&=&\tilde{\rho}^A \otimes \tilde{\rho}^{B_1B_2B_3}\notag\\
&=&\tilde{\rho}^A \otimes\left(\frac{1}{2}\ket{\psi^{(a)}_1}\bra{\psi^{(a)}_1}+\frac{1}{2}\ket{\psi^{(a)}_2}\bra{\psi^{(a)}_2}\right)\notag\\
&\propto& \tilde{\rho}^A \otimes \mathbb{1}^{\mathrm{End}(V^{(a)})},
\end{eqnarray}
where $\mathbb{1}^{\mathrm{End}(V^{(a)})}$ is the identity operator in $\mathrm{End}(V^{(a)})$.

It is straightforward to check that $V^{(a)}$ is an invariant space under permutation group $S_3$, thus the representation on $V^{(a)}$ must be irreducible.
Eq.(\ref{symmetric}) essentially require that each group element of $S_3$ must commute with $\tilde{\rho}^{B_1B_2B_3}$.
By Schur's Lemma, it must be proportional to $\mathbb{1}^{\mathrm{End}(V^{(a)})}$.

Likewise, one can write down a density matrix that purely supports on $\mathrm{End}(V_A)\otimes \mathrm{End}(V^{(b)})$
\begin{eqnarray}
\bar{\rho}^{AB_1B_2B_3}
&=&\bar{\rho}^A \otimes\left(\frac{1}{2}\ket{\psi^{(b)}_1}\bra{\psi^{(b)}_1}+\frac{1}{2}\ket{\psi^{(b)}_2}\bra{\psi^{(b)}_2}\right)\notag\\
&=&\bar{\rho}^A \otimes \bar{\rho}^{B_1B_2B_3},
\end{eqnarray}
where
\begin{subequations}
\begin{eqnarray}
\ket{\psi^{(b)}_1} &\equiv& \frac{1}{\sqrt{6}}\left(2\ket{110}+\ket{101}-\ket{011}\right),\\
\ket{\psi^{(b)}_2} &\equiv& \frac{1}{\sqrt{2}}\left(\ket{101}-\ket{011}\right).
\end{eqnarray}
\end{subequations}

Examining the density matrix supporting purely on the bosonic subspace, one could find that there might exist cross terms like $\ket{\psi}_A\bra{\psi'}_A\otimes\ket{000}_{B_1B_2B_3}\bra{111}_{B_1B_2B_3}$.
Thus it is reasonable to assume that, cross terms mapping from $V^{(a)}$ to $V^{(b)}$ would also exist and vice versa.

After calculation, one could verify that, an Hermitian cross term satisfying Eq.(\ref{symmetric}) must be of the following form
\begin{eqnarray}
\hat{\rho}^{AB_1B_2B_3}&=&\hat{\rho}^A\otimes \hat{\rho}^{B_1B_2B_3}+h.c.,
\end{eqnarray}
where
\begin{equation}
\hat{\rho}^{B_1B_2B_3}=\frac{1}{2}\ket{\psi^{(a)}_1}\bra{\psi^{(b)}_1}+\frac{1}{2}\ket{\psi^{(a)}_2}\bra{\psi^{(b)}_2}.
\end{equation}

A general density matrix $\rho^{AB_1B_2B_3}$ supporting on $\mathrm{End}(V_A)\otimes \mathrm{End}(V^{(a)}\oplus V^{(b)})$ should be a linear combination of $\tilde{\rho}^{AB_1B_2B_3}$, $\bar{\rho}^{AB_1B_2B_3}$ and $\hat{\rho}^{AB_1B_2B_3}$\footnote{Of cause, such coefficients have to satisfy some constrains to insure that $\rho^{AB_1B_2B_3}$ is a legal density matrix.}
\begin{equation}
\rho^{AB_1B_2B_3}=\alpha \tilde{\rho}^{AB_1B_2B_3}+ \beta \bar{\rho}^{AB_1B_2B_3}+ \gamma \hat{\rho}^{AB_1B_2B_3}.
\end{equation}

Define
\begin{subequations}
\begin{eqnarray}
\tilde{\sigma}^{B_1B_2B_3} &\equiv& \ket{\phi_1}_{B_1B_2B_3}\bra{\phi_1}_{B_1B_2B_3},\\
\bar{\sigma}^{B_1B_2B_3} &\equiv& \ket{\phi_2}_{B_1B_2B_3}\bra{\phi_2}_{B_1B_2B_3},\\
\hat{\sigma}^{B_1B_2B_3} &\equiv& \ket{\phi_1}_{B_1B_2B_3}\bra{\phi_2}_{B_1B_2B_3}+h.c.,
\end{eqnarray}
\end{subequations}
where
\begin{subequations}
\begin{eqnarray}
\ket{\phi_1}_{B_1B_2B_3}&\equiv&\frac{1}{\sqrt{3}}\left(\ket{001}+\ket{010}+\ket{100}\right),\\
\ket{\phi_2}_{B_1B_2B_3}&\equiv&\frac{1}{\sqrt{3}}\left(\ket{110}+\ket{101}+\ket{011}\right).
\end{eqnarray}
\end{subequations}

It is easily to check that, a bosonic global state $\sigma^{AB_1B_2B_3}$, which satisfies
\begin{equation}
\Tr_{(AB_1)^c}\sigma^{AB_1B_2B_3}=\Tr_{(AB_1)^c}\rho^{AB_1B_2B_3}=\rho^{AB_1},
\end{equation}
can be obtained by a simple replacement
\begin{subequations}
\begin{eqnarray}
\tilde{\rho}^{B_1B_2B_3} \leftrightarrow \tilde{\sigma}^{B_1B_2B_3},\\
\bar{\rho}^{B_1B_2B_3} \leftrightarrow \bar{\sigma}^{B_1B_2B_3},\\
\hat{\rho}^{B_1B_2B_3} \leftrightarrow \hat{\sigma}^{B_1B_2B_3},
\end{eqnarray}
\end{subequations}
and a coefficient modification $\gamma' \leftrightarrow -\gamma$.\footnote{By multiplying a global phase on basis in $V^{(b)}$, we could always absorb this minus sign or any other phase.}
Clearly, as long as $\rho^{AB_1B_2B_3}$ is a density matrix, which means it has to be positive definite, normalized and Hermitian, $\sigma^{AB_1B_2B_3}$ must also be a proper density matrix.

Though the idea of proof for $k=3$ is quite similar as $k=2$, discrepancies mentioned in section \ref{sec:2} leaded to a more complicated version.

Our proof for general $k$ will also be divided into $3$ steps.

\begin{enumerate}
\item At first we write down a general matrix (not necessarily a density matrix) after symmetric extension.
This step could be further divided into $2$ steps.
Firstly, we shall write down a general matrix in $\ch\equiv\mathrm{End}(V_A)\otimes\mathrm{End}(\otimes^k_{i=1} V_{B_i})$.
The key point of this step is to find an orthogonal and complete basis in $\ch$.
Such basis should be able to conveniently describe the permutation symmetry.
Secondly, we shall restrict such general matrix form according to Eq.(\ref{symmetric}).
As one could expect, not only there exist the diagonal terms that represent mapping inside irreducible subspaces, cross terms that describe mapping between different irreducible subspaces also arise, as long as the representation on both irreducible subspace are equivalent.

\item Then, we shall verify that the former part will become the diagonal terms after partial trace, while the latter one contributes to the off-diagonal terms.
Under our specific situation that $B_1$ is a qubit, only $1$ independent off-diagonal term survives, thus cross term can always be replaced with a bosonic version by properly modifying coefficients.
On the other hand, the ratio between diagonal terms is always same, regardless of whether they are obtained from a bosonic extension or not.

\item The last piece is to demonstrate that the global matrix, obtained by replacing non-bosonic entries with bosonic ones, is positive semi-definite.
\end{enumerate}

\section{The case of $k$-extension}
\label{sec:4}
In this section we will prove the follow main result.
\begin{theorem}\label{main theorem}
For any $k$-extendible state $\rho^{AB}$, if $d_B=2$, then a $k$-bosonic extension always exists.
\end{theorem}

Consider a Hilbert space $\ct=\bigotimes_{i=1}^k V^{(i)}$ constituted by $k$-qubit $B_1, B_2,\cdots, B_k$, whose computational basis is $\left\{\Phi_{i_1,i_2,\cdots,i_k}\equiv\ket{i_1,i_2,\cdots,i_k}\right\}$, where $i_1,i_2,\cdots,i_k=0,1$.

Each subsystem $V^{(i)}$ is invariant under $SU(2)$ rotation, and transforms according to the $2$-dimensional irreducible representation $D^{(2)}$.
Therefore the Lie algebra $su(2)$, which describes the infinite small rotation of $SU(2)$, has the following matrix form on each $V^{(i)}$
\begin{equation*}
J^{(i)}_z=\frac{1}{2}\left(
\begin{tabular}{c c}
1&0 \\
0&-1
\end{tabular}
\right),
\qquad
J^{(i)}_+=\left(
\begin{tabular}{c c}
0&1 \\
0&0
\end{tabular}
\right),
\qquad
J^{(i)}_-=\left(
\begin{tabular}{c c}
0&0 \\
1&0
\end{tabular}
\right),
\end{equation*}
where we have set
\begin{equation*}
\ket{1}^{(i)}=\left(
\begin{tabular}{c}
1\\
0
\end{tabular}
\right),\qquad
\ket{0}^{(i)}=\left(
\begin{tabular}{c}
0\\
1
\end{tabular}
\right),
\end{equation*}

$\ct$ is also invariant under global $SU(2)$ rotation, whose corresponding $su(2)$ algebra is given by $\mathbf{J}_z\equiv\sum_i J^{(i)}_z, \mathbf{J}_{\pm}\equiv\sum_i J^{(i)}_{\pm}$.
$\ct$ transforms under representation $\otimes^k D^{(2)}$, which is not irreducible, but can be decomposed as direct sum of a series of irreducible representations
\begin{eqnarray}\label{multiplicity}
\bigotimes\nolimits^k D^{(2)}=\bigoplus_{j}m_j D^{(2j+1)},
\end{eqnarray}
where $m_j$ is the multiplicity of irreducible representation $D^{(2j+1)}$.
This is equivalent to say that $\ct$ can be partitioned as direct sum of orthogonal subspaces
\begin{equation}
\ct=\bigoplus_{j} m_j \ct^{(2j+1)}.
\end{equation}

In Appendix \ref{SU symmetry}, we manifest that, such $\ct^{(2j+1)}$ has particular permutation symmetry described by Young diagram $[\lambda]$.\footnote{Here $[\lambda]\equiv \{\lambda_1,\lambda_2,\cdots,\lambda_n\}$ is a partition of integer $k$, where all $\lambda_i$ are integers satisfying $\lambda_1\geq \lambda_2\geq\cdots\geq\lambda_n \geq 0, \sum_{i=1}^n\lambda_i=k$. Such partition describes an $n$-row Young diagram.}

Since irreducible representation can be labeled by partition $[\lambda]$.
We can rewrite Eq.(\ref{multiplicity}) with new notation
\begin{eqnarray}\label{multiplicity1}
\bigotimes\nolimits^{k}D^{[1]}=\bigoplus_{[\lambda]}C_k^{[\lambda]}D^{[\lambda]},
\end{eqnarray}
where $C_k^{[\lambda]}$ is the multiplicity of $SU(2)$ irreducible representation $D^{[\lambda]}$.

Two irreducible representation spaces $\ct^{[\lambda]}_{\mu}$ and $\ct^{[\lambda]}_{\nu}$ corresponding to same Young diagram but different Young tableaus are orthogonal to each other.
It is also known that there is no multiplicity in any weight subspace in a $SU(2)$ representation, as long as it is irreducible.
Thus one can safely use weight $\omega$, the eigenvalue of $\mathbf{J}_z$, to label different state inside an irreducible subspace $\ct^{[\lambda]}_{\mu}$.
Therefore, $\{\ket{[\lambda],\mu,\omega}\}$ labels a complete basis of $\ct$ one by one, where $[\lambda]$ tells inequivalent $SU(2)$ representations while $\mu$ differentiate equivalent ones.
They together determine an orthogonal irreducible subspace, and $\omega$ labels every different vectors inside.

On the other hand, $\{\ket{[\lambda],\mu,\omega}\}$ can be interpreted in another way:
$\omega$ describes the weight, $[\lambda]$ tells inequivalent $S_k$ representations, thus these two parameter differentiate orthogonal invariant subspaces, while $\mu$ labels vectors inside.\footnote{The validity of this explanation is verified in Appendix \ref{Sn symmetry}.}
From now on we shall use $\ket{\omega^{[\lambda]}_{\mu}}$ short for $\ket{\omega,[\lambda],\mu}$.

Any matrix $\rho^{A B_1 B_2 \cdots B_n} \in \mathrm{End}(V_A) \otimes \mathrm{End}(\ct)$ can be expressed as
\begin{eqnarray}\label{general_form}
\rho^{A B_1 B_2 \cdots B_k}&=&\sum_{[\lambda],[\lambda']}\sum_{\mu,\mu'}\sum_{\omega,\omega'}\sum_{\alpha,\alpha'}
\ket{\psi^{\alpha}_{\omega,[\lambda],\mu}}\bra{\psi^{\alpha'}_{\omega',[\lambda'],\mu'}}\notag\\
&&\qquad \otimes\ket{\omega^{[\lambda]}_{\mu}}\bra{\omega'^{[\lambda']}_{\mu'}},\notag\\
\end{eqnarray}
where $\ket{\psi^{\alpha}_{\omega,[\lambda],\mu}}$ are non-normalized state in $V_A$ and $\alpha$ label different states in $V_A$.

Insert Eq.(\ref{general_form}) into Eq.(\ref{symmetric}). $\forall \pi \in S_k$ we get a series of constrains for $\rho^{A B_1 B_2 \cdots B_k}$:
\begin{eqnarray}
&&\qquad\forall [\lambda],[\lambda']\omega,\omega'\textrm{ and }\mu,\mu'\notag\\
&&\sum_{\alpha,\alpha'}\ket{\psi^{\alpha}_{\omega,[\lambda],\mu}}\bra{\psi^{\alpha'}_{\omega',[\lambda'],\mu'}}
\sum_{\nu,\nu'}
\ca(\pi)^{[\lambda]}_{\mu,\nu}
\ca(\pi)^{[\lambda']*}_{\nu',\mu'}
\ket{\omega^{[\lambda]}_{\nu}}
\bra{\omega'^{[\lambda']}_{\nu'}}\notag\\
&&\qquad=
\sum_{\alpha,\alpha'}\ket{\psi^{\alpha}_{\omega,[\lambda],\mu}}\bra{\psi^{\alpha'}_{\omega',[\lambda'],\mu'}}
\ket{\omega^{[\lambda]}_{\mu}}\bra{\omega'^{[\lambda']}_{\mu'}},\qquad
\end{eqnarray}
where $\ca^{[\lambda]}$ and $\ca^{[\lambda']}$ are irreducible representations of permutation group $S_k$.

Define matrix
\begin{eqnarray}
&&M(\omega,\omega',[\lambda],[\lambda'])\nonumber\\
&\equiv&
\sum_{\mu,\mu'}M(\omega,\omega',[\lambda],[\lambda'])_{\mu\mu'}
\ket{\omega^{[\lambda]}_{\mu}}\bra{\omega'^{[\lambda']}_{\mu'}},
\end{eqnarray}
where
\begin{equation}
M(\omega,\omega',[\lambda],[\lambda'])_{\mu\mu'} \equiv\sum_{\alpha,\alpha'} \ket{\psi^{\alpha}_{\omega,[\lambda],\mu}}\bra{\psi^{\alpha'}_{\omega',[\lambda'],\mu'}},
\end{equation}
thus $\forall \pi \in S_k$
\begin{eqnarray}
\label{commute}
&&\ca^{[\lambda]}(\pi)M(\omega,\omega',[\lambda],[\lambda'])\ca^{[\lambda']}(\pi)^{\dagger}\nonumber\\
&=&M(\omega,\omega',[\lambda],[\lambda']).\notag\\
\end{eqnarray}

Schur's lemma guarantee that,
\begin{itemize}
\item[a.] when $[\lambda] \neq [\lambda']$, $M=0$;
\item[b.] when $[\lambda] = [\lambda']$, $M$ is invertible.
\end{itemize}

Choose $\ket{\omega_{\mu}^{[\lambda]}}$ carefully such that the representation $\ca^{[\lambda]}$ are identical, not just isomorphic, for different weight $\omega$.
Hence all $M(\omega,\omega',[\lambda],[\lambda])$ can be proportional to the corresponding identity matrix.
Therefore, one could eliminate plenty of cross terms and restrict $\rho^{A B_1 B_2 \cdots B_k}$ to
\begin{eqnarray}\label{specific form}
\rho^{A B_1 B_2 \cdots B_k}&=&\sum_{[\lambda]}\sum_{\omega,\omega'}
\left(\sum_{\alpha,\alpha'}\ket{\psi^{\alpha}_{\omega,[\lambda]}}\bra{\psi^{\alpha'}_{\omega',[\lambda]}}\right)\notag\\
&&\otimes\frac{1}{d^{[\lambda]}}\sum_{\mu}\ket{\omega^{[\lambda]}_{\mu}}\bra{\omega'^{[\lambda]}_{\mu}},
\end{eqnarray}
where $d^{[\lambda]}$ is the dimension of $S_k$ irreducible representation corresponding to Young diagram $[\lambda]$.

Determining RDM $\rho^{AB_1}$ can be immediately reduced to calculating every possible combination of $[\lambda], \omega$ and $\omega'$.
For given $[\lambda],\omega$ and $\omega'$, one could temporally ignore system $A$ and concentrate on group $\{B_1,B_2,\cdots,B_k\}$.
Then the task left is to calculate
\begin{eqnarray}\label{partial}
\frac{1}{d^{[\lambda]}}\Tr_{(B_1)^c}\sum_{\mu}\ket{\omega^{[\lambda]}_{\mu}}\bra{\omega'^{[\lambda]}_{\mu}}.
\end{eqnarray}

If $\omega=\omega'$, it is equivalent to consider a mixed state within a constant weight subspace.
Hence
\begin{eqnarray}
\frac{1}{d^{[\lambda]}}\Tr_{(B_1)^c}\sum_{\mu}\ket{\omega^{[\lambda]}_{\mu}}\bra{\omega^{[\lambda]}_{\mu}}=t_0\ket{0}\bra{0}+t_1\ket{1}\bra{1}.
\end{eqnarray}
According to \cite{Chen2017Local}\footnote{see details in Appendix \ref{diag}},
\begin{equation}
\frac{t_0}{t_1}=\frac{k-2\omega}{k+2\omega}.
\end{equation}

Since the ratio between diagonal terms is solely determined by the number of subsystems $k$ and weight $\omega$.
Any non-bosonic extensions can be directly replaced by a bosonic version in same weight subspace
\begin{eqnarray}
\frac{1}{d^{[\lambda]}}\sum_{\mu}\ket{\omega^{[\lambda]}_{\mu}}\bra{\omega^{[\lambda]}_{\mu}}
\leftrightarrow
\ket{\omega^S}\bra{\omega^S}
\end{eqnarray}

If $\omega-\omega'=\pm 1$, nonzero contribution of Eq.(\ref{partial}) would be proportional to $\ket{1}\bra{0}$ and $\ket{0}\bra{1}$ respectively. Different $[\lambda]$ only affect the proportion coefficients.
Choose proper coefficients for bosonic extension will exactly recover the result of non-bosonic ones,\footnote{See details in Appendix \ref{off diag 1}.}
\begin{eqnarray}
\frac{1}{d^{[\lambda]}}\sum_{\mu}\ket{\omega^{[\lambda]}_{\mu}}\bra{\omega'^{[\lambda]}_{\mu}}
\leftrightarrow
\beta^{[\lambda]}_{\omega,\omega'}\ket{\omega^{S}}\bra{\omega'^{S}}
\end{eqnarray}
where
\[\beta^{[\lambda]}_{\omega,\omega'}=\sqrt{
\frac{(\frac{\lambda_1-\lambda_2}{2}-\omega)(\frac{\lambda_1-\lambda_2}{2}+\omega+1)}
{(\frac{k}{2}-\omega)(\frac{k}{2}+\omega+1)}}\delta_{\omega \pm 1,\omega'}.
\]
Obviously, $0 < \beta^{[\lambda]}_{\omega,\omega'} \leq 1$.

If $|\omega-\omega'|\geq 2$, Eq.(\ref{partial}) would vanish.

According to Eq.(\ref{specific form}), $\rho^{A B_1 B_2 \cdots B_k}$ has a series of bosonic version $\sigma^{A B_1 B_2 \cdots B_k}$ (not all of them are proper density matrix)
\begin{eqnarray}\label{positive}
\sigma^{A B_1 B_2 \cdots B_k} &\equiv & \sum_{[\lambda]}\sum_{\omega,\omega'}
\left(\sum_{\alpha,\alpha'}\ket{\psi^{\alpha}_{\omega,[\lambda]}}\bra{\psi^{\alpha'}_{\omega',[\lambda]}}\right)\notag\\
&&\otimes p^{[\lambda]}_{\omega,\omega'} \ket{\omega^{S}}\bra{\omega'^{S}},
\end{eqnarray}
where $p^{[\lambda]}_{\omega,\omega'}$ are coefficients
\begin{eqnarray}\label{p-matrix}
p^{[\lambda]}_{\omega,\omega'}=
\left\{
\begin{tabular}{l c}
1, & $\omega=\omega'$\\
$\sqrt{\frac{(\frac{\lambda_1-\lambda_2}{2}-\omega)(\frac{\lambda_1-\lambda_2}{2}+\omega+1)}
{(\frac{k}{2}-\omega)(\frac{k}{2}+\omega+1)}}$, & $\omega\pm 1=\omega'$\\
arbitrary value, & others
\end{tabular}
\right.
\end{eqnarray}

We can always find proper $p^{[\lambda]}_{\omega,\omega'}$ such that $\sigma^{A B_1 B_2 \cdots B_k}$ can be decomposed as a convex combination of a series of pure bosonic extensions.\footnote{See Appendix \ref{off diag 2} for details.}
Hence $\sigma^{A B_1 B_2 \cdots B_k}$ is positive definite.

Therefore we have finished the proof of Theorem \ref{main theorem}.

\section{discussion}
\label{sec:5}

We have shown that if a bipartite state $\rho^{AB}$ has a $k$-symmetric extension
\begin{eqnarray}
\left(\mathbb{1}^A\otimes \pi\right) \rho^{AB_1B_2\cdots B_k} \left(\mathbb{1}^A\otimes \pi\right)^{\dagger}=\rho^{AB_1B_2\cdots B_k}
\end{eqnarray}
with $\rho^{AB_i}=\rho^{AB}$,
it must also have a bosonic extension $\sigma^{AB_1B_2\cdots B_k}$ satisfying
\begin{subequations}
\begin{eqnarray}
&&\sigma^{AB_1B_2\cdots B_k}=\sum_{\alpha}p_{\alpha}\ket{\phi_{\alpha}}^{AB_1B_2\cdots B_k}\bra{\phi_{\alpha}}^{AB_1B_2\cdots B_k},\\
&&\left(\mathbb{1}^A\otimes \pi\right) \ket{\phi_{\alpha}}^{AB_1B_2\cdots B_k} = \ket{\phi_{\alpha}}^{AB_1B_2\cdots B_k},
\end{eqnarray}
\end{subequations}
where $\pi \in S_k$ is an arbitrary permutation operator, and $p_{\alpha}$ is a probability distribution as long as $B_1$ is a qubit.

Notice that, Eq.(\ref{specific form}) was essentially saying that $\rho^{AB_1B_2\cdots B_k}$ could be further decomposed into $2$ major parts. The first part contributed to the ``diagonal terms'', whose ratio is identical as long as global state lied in same weight subspaces.
Hence this part did not contain information about permutation symmetry.
The second type contributed to the ``off-diagonal terms''.
They probably carried information about permutation symmetry.

Since we have discussed the case for $k$-qubit extension, it is natural to consider the $k$-qudit problem, i.e. $d_B>2$.
However, it turns out that the situation is much more complicated, since there are now more than $1$ pair of off-diagonal terms after partial trace.
In the qubit case, we can peel off all the off-diagonal terms at one time for each given $[\lambda]$.
Due to the properties of $\beta^{[\lambda]}_{\omega,\omega'}$, the residual ``diagonal'' matrix is guaranteed to be positive definite.
In the qudit case, we might have to peel off the off-diagonal terms in several steps.
After peeling off all off-diagonal terms, the residual might not be always positive definite.
\begin{example}
Consider a tripartite pure state of on $V_A\otimes V_{B_1}\otimes V_{B_2}$ the form
\begin{eqnarray}
\ket{\psi}&=&\alpha(\ket{012}-\ket{021})\nonumber\\
&+&\beta(\ket{120}-\ket{102})+\gamma(\ket{201}-\ket{210}),
\end{eqnarray}
where $\alpha, \beta, \gamma \neq 0$ are all different.
Clearly, $\ket{\psi}$ is a fermionic extension of $\rho^{AB_1}\equiv \Tr_{B_2}[\ket{\psi}\bra{\psi}]$.
The diagonal terms of $\bar{\rho}^{AB_1}$ are
\begin{eqnarray}
\notag\bar{\rho}^{AB_1}&=&\alpha\alpha^*(\ket{01}\bra{01}+\ket{02}\bra{02})\\
\notag&&+\beta\beta^*(\ket{12}\bra{12}+\ket{10}\bra{10})\\
&&+\gamma\gamma^*(\ket{20}\bra{20}+\ket{21}\bra{21}),
\end{eqnarray}
while the off-diagonal $\tilde{\rho}^{AB_1}$ read
\begin{eqnarray}
\notag&&\tilde{\rho}^{AB_1}=-\alpha\beta^*\ket{01}\bra{10}-\alpha\gamma^*\ket{02}\bra{20}\\
&-&\beta\gamma^*\ket{12}\bra{21}+h.c.
\end{eqnarray}
If $\rho^{AB_1}$ had a bosonic extension $\sigma^{AB_1B_2}$, then $\Tr_{B_2}\sigma^{AB_1B_2}$ should produce exact off-diagonal terms as $\tilde{\rho}^{AB_1}$.

$-\alpha\beta^*\ket{01}\bra{10}$ can be obtained in $3$ different ways\footnote{of course it can be a mixture of these $3$}
\begin{subequations}
\begin{eqnarray}
&& \Tr_{B_2}[\alpha\beta^* \ket{0}\bra{1}\otimes (\ket{10}+\ket{01})\bra{00}],\\
&& \Tr_{B_2}[\alpha\beta^* \ket{0}\bra{1}\otimes \ket{11}(\bra{10}+\bra{01})],\\
&& \Tr_{B_2}[\alpha\beta^* \ket{0}\bra{1}\otimes (\ket{12}+\ket{21})(\bra{02}+\bra{20})]\label{alphabeta}.
\end{eqnarray}
\end{subequations}
However, in order to keep $\sigma^{AB_1B_2}$ positive definite, the first two choices will have to introduce $\ket{00}\bra{00}$ and $\ket{11}\bra{11}$ respectively, which did not appear in diagonal terms, and hence destroy the positive definiteness.
Therefore, we have to use Eq.(\ref{alphabeta}) to obtain $-\alpha\beta^*\ket{01}\bra{10}$ in $\Tr_{B_2}\sigma^{AB_1B_2}$.

After replacing all terms in $\tilde{\rho}^{AB_1}$ with corresponding bosonic extension, the off-diagonal terms in $\sigma^{AB_1B_2}$ are
\begin{eqnarray}
\notag\tilde{\sigma}^{AB_1B_2}&=& -\bigg\{\alpha\beta^*\ket{0}\bra{1}\otimes(\ket{12}+\ket{21})(\bra{02}+\bra{20})\\
\notag&&+\alpha\gamma^*\ket{0}\bra{2}\otimes (\ket{12}+\ket{21})(\bra{01}+\bra{10})\\
&&\notag+\beta\gamma^*\ket{1}\bra{2}\otimes (\ket{20}+\ket{02})(\bra{01}+\bra{10})\\
&& +h.c.\bigg\}.
\end{eqnarray}

Because of the global minus sign, it is impossible to peel off all the $3$ pairs off-diagonal terms or any two of them at one time.
In other words, we have to peel off them pair by pair.
After matching the corresponding diagonal terms with off-diagonal pair, the ``residual'' diagonal matrix in $\Tr_{B_3}$ will be
\begin{eqnarray}
\notag &&(\alpha\alpha^*-pp^*-ss^*)(\ket{01}\bra{01}+\ket{02}\bra{02})\\
\notag &&+(\beta\beta^*-qq^*-uu^*)(\ket{12}\bra{12}+\ket{10}\bra{10})\\
&&+(\gamma\gamma^*-tt^*-vv^*)(\ket{20}\bra{20}+\ket{21}\bra{21}),
\end{eqnarray}
where
\begin{eqnarray}
pq^*=\alpha\beta^*,\quad st^*=\alpha\gamma^*,\quad uv^*=\beta\gamma^*.
\end{eqnarray}
It is easy to check that
\begin{eqnarray}
\notag&&\alpha\alpha^*-pp^*-ss^*+\beta\beta^*-qq^*\\
&&-uu^*+\gamma\gamma^*-tt^*-vv^* < 0.
\end{eqnarray}
Therefore, there does not exist bosonic extension for $\rho^{AB_1}$.
\end{example}

However, there may be some coincident situation, under which the ``residual'' diagonal matrix is positive definite, hence the symmetric extendibility of $\rho^{AB_1}$ is equivalent to its bosonic extendibility.
We do not know whether such coincidences are purely accidental or there are some underlining profound reasons. We leave this for future research.

\section*{Acknowledgments}
We thank Jianxin Chen, Zhengfeng Ji and Nengkun Yu for helpful discussions. YL acknowledges support from the Chinese Ministry of Education under grant No. 20173080024.
BZ is supported by Natural Science and Engineering Research Council of Canada (NSERC) and Canadian Institute for Advanced Research (CIFAR).

\numberwithin{equation}{section}
\appendix

\section{permutation symmetry of irreducible subspace $\ct^{(2j+1)}$}\label{SU symmetry}

We define $P^{mn}$, an element in permutation group $S_k$ as
\begin{eqnarray}
P^{mn}\Phi_{i_1,\cdots,i_m,\cdots,i_n,\cdots,i_k}=\Phi_{i_1,\cdots,i_n,\cdots,i_m,\cdots,i_k}
\end{eqnarray}

It is obvious that the permutation of the indices of subsystems commutes with the tensor product of individual $SU(2)$ rotation, and hence the global $SU(2)$ rotation.
Therefore, any subspace
\begin{equation}
\ct^{[\lambda]}_{\mu}\equiv \cy^{[\lambda]}_{\mu}\ct
\end{equation}
projected by a standard Young tableau\footnote{Standard Young tableau is obtained by filling $\{1, 2,\cdots, k\}$ into all entries in such a manner that each row and each column keeps in increasing order.} $\cy^{[\lambda]}_{\mu}$ is also invariant under the global $SU(2)$ rotation.
Here $[\lambda]$ describes an $n$-row Young diagram, and $\mu$ differentiates standard Young tableaus corresponding to same $[\lambda]$.

Since $\{\Phi_{i_1,i_2,\cdots,i_k}\}$ is a complete basis of $\ct$, $\{\cy^{[\lambda]}_{\mu}\Phi_{i_1,i_2,\cdots,i_k}\}$ must be a super-complete basis of $\ct^{[\lambda]}_{\mu}$.
From now on, to describe the vector $\cy^{[\lambda]}_{\mu}\Phi_{i_1,i_2,\cdots,i_k}$, we shall use a graphic way: insert $i_j$ into the entry of Young diagram $\cy^{[\lambda]}$ that is filled by $j$ in standard Young tableau $\cy^{[\lambda]}_{\mu}$
\begin{example}
\Yvcentermath1
\Ystdtext1
\newcommand{\ia}{$i_1$}
\newcommand{\ib}{$i_2$}
\newcommand{\ic}{$i_3$}
\begin{eqnarray}
\young(\ia \ic ,\ib) &\equiv& \young(13,2)\Theta_{i_1,i_2,i_3}\\
& = & (E+ (1 3))(E-(1 2))\Phi_{i_1,i_2,i_3}\notag\\
& = & \Phi_{i_1,i_2,i_3} + \Phi_{i_3,i_2,i_1} -\Phi_{i_2,i_1,i_3} -\Phi_{i_3,i_1,i_2}.\notag
\end{eqnarray}
\end{example}

Because each $V^{(i)}$ is a $2$-dimensional subsystem, and for any given Young tableau, its corresponding Young operator anti-symmetrize each column, Young diagram with $3$ or more rows would project into the empty space, thus from now on we only consider $1$ or $2$-row Young diagram.

Moreover, if $2$ states belong to a same invariant subspace, say $\ct^{[\lambda]}_{\mu}$, are identical except interchanging entries in a given column, their difference would only be a factor $-1$.
\begin{example}
\begin{equation*}
\Yvcentermath1
\young(11,0) = -1 \times \young(01,1).
\end{equation*}
\end{example}
Therefore, if a state fills two $0$s or two $1$s in the same column, such state must vanish.

It can be easily verified that, the global $SU(2)$ representation on $\ct^{[\lambda]}_{\mu}$ is irreducible, by the following observation:
\begin{observation}
There is a single highest weight state $\ket{\omega_M}$ in each $\ct^{[\lambda]}_{\mu}$
\begin{equation}
\Yvcentermath1
\Ystdtext1
\newcommand{\lds}{$\ldots$}
\ket{\omega_M}=\young(111\lds \lds 1,00\lds 0).
\end{equation}
\end{observation}

\begin{proof}
It can be easily verified that the effect of $\mathbf{J}_+$ acting on any state in $\ct$ is just simply lifting $\ket{0}$ to $\ket{1}$ or annihilating $\ket{1}$, then summing all modified states together.
Therefore $\mathbf{J}_+\ket{\omega_M}$ would vanish since each modified state will either be annihilated directly or have two $\ket{1}$s filling in a column.
Thus $\ket{\omega_M}$ is a highest weight state.

States with same amount of $\ket{1}$ and $\ket{0}$ as $\ket{\omega_M}$ would be either $0$ or can be obtained by simply interchanging entries within one or more columns of $\ket{\omega_M}$, which would at most contribute an additional factor $-1$.

States with more $\ket{1}$ than $\ket{\omega_M}$ vanish directly.

States that replace several $\ket{1}$ to $\ket{0}$ in the first row will survive after the action of $\mathbf{J}_+$, thus can not be a highest weight state.

Therefore $\ket{\omega_M}$ indeed is the single highest weight state in $\ct^{[\lambda]}_{\mu}$.
\end{proof}

\section{Construct irreducible subspace of $S_k$}\label{Sn symmetry}
For any element $\pi \in S_k$, $\pi \ct^{[\lambda]}_{\mu}$ will either leave $\ct^{[\lambda]}_{\mu}$ untouched or map it integrally to another $\ct^{[\lambda]}_{\nu}$
\begin{equation}
\pi \ct^{[\lambda]}_{\mu} = \pi \cy^{[\lambda]}_{\mu}\ct=\cy^{[\lambda]}_{\nu} \pi\ct=\ct^{[\lambda]}_{\nu}.
\end{equation}

Furthermore, any element $\pi \in S_k$ can not change the state weight.
Thus states with same weight $\omega_0$, irreducible label $[\tilde{\lambda}]$, but different $\mu$ span an invariant subspace $\ct^{[\tilde{\lambda}]}(\omega_0)$ under permutation.
The representation, say $\ca$, on $\ct^{[\tilde{\lambda}]}(\omega_0)$ can be decomposed as a direct sum of irreducible representations of $S_k$, and so does $\ct^{[\tilde{\lambda}]}(\omega_0)$ itself.

Any state in $\ct^{[\tilde{\lambda}]}(\omega_0)$ is obtained by a projection operator corresponding to Young diagram $[\tilde{\lambda}]$.
Hence in $\ct^{[\tilde{\lambda}]}(\omega_0)$ only appears irreducible representation described by $\ca^{[\tilde{\lambda}]}$.
Furthermore, it can be quickly obtained that the multiplicity of $\ca^{[\tilde{\lambda}]}$ must be $1$, as a quickly corollary of the following statement:
\begin{observation}
Suppose $C_k^{[\lambda]}$ is the multiplicity of $SU(2)$ irreducible representation $D^{[\lambda]}$ appearing in Eq.(\ref{multiplicity1}), then $C_k^{[\lambda]}$ equals to $d^{[\lambda]}$, the dimension of irreducible representation $\ca^{[\lambda]}$ of \textbf{permutation group} $S_k$.
\end{observation}

\begin{proof}
When $k=1$, the statement is trivial.

Suppose statement holds when $k=m$.

For $k=m+1$,
\begin{eqnarray}
\bigotimes\nolimits^{m+1}D^{[1]}&=&\left(\bigoplus_{\lambda_1+\lambda_2=m}C^{[\lambda]}_m D^{[\lambda]}\right)\bigotimes D^{[1]}\notag\\
&=&\bigoplus_{\lambda'_1+\lambda'_2=m+1}C^{[\lambda']}_{m+1} D^{[\lambda']},
\end{eqnarray}
according to Littlewood-Richardson rule,
\begin{eqnarray}
C_{m+1}^{\{\lambda'_1,\lambda'_2\}}=C_m^{\{\lambda'_1-1,\lambda_2\}}+C_m^{\{\lambda'_1,\lambda'_2-1\}},
\end{eqnarray}
where $C_m^{\{\lambda_1,\lambda_2\}}$ equals to $d^{[\lambda]}$, the dimension of $S_m$ irreducible representation $\ca^{\{\lambda_1,\lambda_2\}}$.
$d^{[\lambda]}$ can be easily computed from its Young diagram by a result known as the hook-length formula
\begin{eqnarray}\label{dimension}
d^{\{\lambda_1,\lambda_2\}}=\frac{(\lambda_1+\lambda_2)!(\lambda_1-\lambda_2+1)}{\lambda_2!(\lambda_1+1)!}.
\end{eqnarray}
By simple calculation, one can verify that
\begin{equation}
C_{m+1}^{\{\lambda'_1,\lambda'_2\}}=\frac{(\lambda'_1+\lambda'_2)!(\lambda'_1-\lambda'_2+1)}{\lambda'_2!(\lambda'_1+1)!}=d^{\{\lambda'_1,\lambda'_2\}},
\end{equation}
and thus the proof is accomplished.
\end{proof}

This indicates that an irreducible subspace of permutation group $S_k$ in $\ct$ can be uniquely determined by weight $\omega$ and irreducible representation label $[\lambda]$.

\section{ratio between diagonal terms}\label{diag}
The combinatorial properties of the constant weight condition impose strong constraints on the reduced density matrices\cite{Chen2017Local}
\begin{subequations}
\begin{eqnarray}
\sum_i(k-2\omega-2)\rho^{B_1B_i}_{01,01}&=&\sum_i(k+2\omega)\rho^{B_1B_i}_{00,00},\\
\sum_i(k+2\omega-2)\rho^{B_1B_i}_{10,10}&=&\sum_i(k-2\omega)\rho^{B_1B_i}_{11,11}.
\end{eqnarray}
\end{subequations}
Assume that
\begin{subequations}
\begin{eqnarray}
\sum_i\rho^{B_1B_i}_{00,00}&=&(k-2\omega-2)t_0,\\
\sum_i\rho^{B_1B_i}_{10,10}&=&(k-2\omega)t_1.
\end{eqnarray}
\end{subequations}
Then we can express the matrix elements using $t_0$ and $t_1$
\begin{subequations}
\begin{eqnarray}
\rho^{B_1}_{0,0}&=&\frac{1}{k-1}\sum_i\left(\rho^{B_1B_i}_{00,00}+\rho^{B_1B_i}_{01,01}\right)=2t_0,\\
\rho^{B_1}_{1,1}&=&\frac{1}{k-1}\sum_i\left(\rho^{B_1B_i}_{10,10}+\rho^{B_1B_i}_{11,11}\right)=2t_1,\\
\sum_i\rho^{B_i}_{0,0}&=&\sum_i\rho^{B_1B_i}_{00,00}+\sum_i\rho^{B_1B_i}_{10,10}\notag\\
&=&(k-2\omega-2)t_0+(k-2\omega)t_1,\\
\sum_i\rho^{B_i}_{1,1}&=&\sum_i\rho^{B_1B_i}_{00,00}+\sum_i\rho^{B_1B_i}_{01,01}\notag\\
&=&(k+2\omega)t_0+(k+2\omega-2)t_1.
\end{eqnarray}
\end{subequations}
Compatibility of $k$-symmetric extendible state requires that,
\begin{equation}
\rho^{B_1}_{0,0}:\rho^{B_1}_{1,1}=\rho^{B_i}_{0,0}:\rho^{B_i}_{1,1},
\end{equation}
i.e.
\begin{equation*}
\frac{2t_0}{2t_1}
=\frac{(k-2\omega-2)t_0+(k-2\omega)t_1}{(k+2\omega)t_0+(k+2\omega-2)t_1},
\end{equation*}
which leads to
\begin{equation}
\frac{t_0}{t_1}=\frac{k-2\omega}{k+2\omega},
\end{equation}
here we have dropped the solution $t_0+t_1=0$.

\section{adjusting coefficient from non-bosonic extension to bosonic extension}\label{off diag 1}

Suppose there are two $k$-qubit pure states $\ket{\psi}$ and $\ket{\phi}$, which lie in different constant weight subspaces $V_{\omega}$ and $V_{\omega+1}$,
\begin{subequations}
\begin{eqnarray}
\ket{\psi} &=& \sum_{i_1,\cdots,i_k=0,1}a_{i_1,\cdots,i_k}\ket{{i_1,\cdots,i_k}},\\
\ket{\phi} &=& \sum_{j_1,\cdots,j_k=0,1}b_{j_1,\cdots,j_k}\ket{{j_1,\cdots,j_k}}.
\end{eqnarray}
\end{subequations}
Define $\rho\equiv\ket{\psi}\bra{\phi}$ and consider its $1$-particle partial trace $\{\rho^{(m)}\equiv\Tr_{m^c}\rho\}$.
It is easy to verify that $\rho^{(m)}$ has only $1$ non-zero element
\begin{eqnarray}
\notag
\rho^{(m)} &=& \sum_{\substack{ i_1,\cdots,i_{m-1},i_{m+1}\cdots,i_k \\ j_1,\cdots,j_{m-1},j_{m+1}\cdots,j_k}} a_{i_1,\cdots,i_k}b_{j_1,\cdots,j_k}\delta_{i_m 0}\delta_{j_m 1}\\
& &\times \delta_{i_1j_1} \cdots \delta_{i_{m-1}j_{m-1}}\delta_{i_{m+1}j_{m+1}}\cdots\delta_{i_{k}j_{k}}\ket{0}\bra{1}.
\end{eqnarray}

$J^{(m)}_+$ acts on the $m$th particle, elevating $\ket{0}$ to $\ket{1}$ and annihilating the  $\ket{1}$, hence
\begin{eqnarray}
\notag\Tr(J^{(m)}_+\rho) &=& \sum_{\substack{ i_1,\cdots,i_{m-1},i_{m+1}\cdots,i_k \\ j_1,\cdots,j_{m-1},j_{m+1}\cdots,j_k}} a_{i_1,\cdots,i_k}b_{j_1,\cdots,j_k}\delta_{i_m 0}\delta_{j_m 1}\\
& &\times \delta_{i_1j_1} \cdots \delta_{i_{m-1}j_{m-1}}\delta_{i_{m+1}j_{m+1}}\cdots\delta_{i_{k}j_{k}}.
\end{eqnarray}

Therefore $\rho^{(m)}=\Tr(J^{(m)}_+\rho)\ket{0}\bra{1}$. Moreover,
\begin{equation}
\sum_{m}\rho^{(m)}=\Tr(\mathbf{J}_+\rho)\ket{0}\bra{1}.
\end{equation}

Now recall task in Eq.(\ref{partial}), one immediately obtain
\begin{eqnarray}
&&\frac{1}{d^{[\lambda]}}\Tr_{(B_1)^c}\sum_{\mu}\ket{\omega^{[\lambda]}_{\mu}}\bra{\omega'^{[\lambda]}_{\mu}}\notag\\
&=&\frac{1}{kd^{[\lambda]}}\sum_{i}\Tr_{(B_i)^c}\sum_{\mu}\ket{\omega^{[\lambda]}_{\mu}}\bra{\omega'^{[\lambda]}_{\mu}}\notag\\
&=&\frac{1}{kd^{[\lambda]}}\Tr\left(\mathbf{J}_+\sum_{\mu}\ket{\omega^{[\lambda]}_{\mu}}\bra{\omega'^{[\lambda]}_{\mu}}\right)\ket{0}\bra{1}.
\end{eqnarray}
Since for each given $\mu$, all possible $\{\ket{\omega^{[\lambda]}_{\mu}}\}$ forms the same irreducible $SU(2)$ representation $D^{[\lambda]}$, corresponding to irreducible representation $j=\frac{1}{2}(\lambda_1-\lambda_2)$.
Within such representation, $\mathbf{J}_+$ elevates $\ket{\omega}$ to $\ket{\omega+1}$, for every weight $\omega$ that is not the highest weight.
\begin{eqnarray}
\frac{1}{d^{[\lambda]}}\Tr_{(B_1)^c}\sum_{\mu}\ket{\omega^{[\lambda]}_{\mu}}\bra{\omega'^{[\lambda]}_{\mu}}
=\alpha^{[\lambda]}_{\omega,\omega'}\,\ket{0}\bra{1},
\end{eqnarray}
where
\[
\alpha^{[\lambda]}_{\omega,\omega'}=\frac{\delta_{\omega+1,\omega'}}{k}
\sqrt{(\frac{\lambda_1-\lambda_2}{2}-\omega)(\frac{\lambda_1-\lambda_2}{2}+\omega+1)}.
\]

\section{density matrix given in Eq.(\ref{positive}) could be positive definite}\label{off diag 2}
Notice that $\sigma^{AB_1B_2\cdots B_k}$ in Eq.(\ref{positive}) is a convex combination of different $[\lambda]$ ingredients:
\begin{eqnarray}
\sigma^{A B_1 B_2 \cdots B_k}=\sum_{[\lambda]}\sigma^{AB_1B_2\cdots B_k}_{[\lambda]},
\end{eqnarray}
where
\begin{eqnarray}
&&\sigma^{AB_1B_2\cdots B_k}_{[\lambda]} \equiv \sum_{\omega,\omega'}
\left(\sum_{\alpha,\alpha'}
|\psi^{\alpha}_{\omega,[\lambda]}\rangle\langle\psi^{\alpha'}_{\omega',[\lambda]}|\right)\notag\\
&&\otimes p^{[\lambda]}_{\omega,\omega'} |\omega^{S}\rangle\langle\omega'^{S}|,
\end{eqnarray}
here $p^{[\lambda]}_{\omega,\omega'}$ is defined in Eq.(\ref{p-matrix}).
Therefore, it is sufficient to verify that $\sigma^{AB_1B_2\cdots B_k}_{[\lambda]}$ can be a positive definite matrix.

We give a construction as below:
\begin{eqnarray}\label{construction}
\sigma^{A B_1 B_2 \cdots B_k}_{[\lambda]}&=&
\left(\sum_{\omega}\xi^{[\lambda]}_{\omega}|\omega^S\rangle\right)
\left(\sum_{\omega'}\xi^{*[\lambda]}_{\omega'}\langle\omega'^S|\right)
\notag\\
&&\,\,+\sum_{\omega}(1-|\xi^{[\lambda]}_{\omega}|^2)
|\omega^S\rangle\langle\omega^S|,
\end{eqnarray}
where
\begin{eqnarray}
\xi^{[\lambda]}_{\omega}\xi^{*[\lambda]}_{\omega+1} &  =   &   p^{[\lambda]}_{\omega,\omega+1},\label{recursion}\\
|\xi^{[\lambda]}_{\omega}|^2  & \leq &   1.\label{constrain}
\end{eqnarray}
The first term of right hand side in Eq.(\ref{construction}) contributes all the off-diagonal terms of $\sigma^{A B_1 B_2 \cdots B_k}_{[\lambda]}$ and the second term consists of purely diagonal terms. The rest part is to give an explicit example of $\{\xi^{[\lambda]}_{\omega}\}$.

Set $\xi^{[\lambda]}_{\frac{k}{2}+1}=1$ and $\xi^{[\lambda]}_{\frac{k+1}{2}}=\xi^{[\lambda]}_{\frac{k+1}{2}}=
\sqrt{p^{[\lambda]}_{\frac{k-1}{2},\frac{k+1}{2}}}$ for even and odd $k$, the rest $\xi^{[\lambda]}_{\omega}$ can be obtained by iteration relation Eq.(\ref{recursion})

Eq.(\ref{constrain}) is satisfied due to the fact that
\begin{itemize}
\item[a.] The maximum value of $p^{[\lambda]}_{\omega,\omega+1}$ is obtained at $\omega_m=\frac{k}{2}+1$ and $\omega_m=\frac{k\pm 1}{2}$ when $k$ is even and odd respectively.
\item[b.] $p^{[\lambda]}_{\omega,\omega+1}$ increases strictly when $\omega < \omega_m$ and decreases strictly when $\omega> \omega_m$.
\end{itemize}

It is straightforward to verify that $\sigma^{A B_1 B_2 \cdots B_k}_{[\lambda]}$ is positive definite as long as its counterpart in the non-bosonic extension is. Therefore, $\sigma^{A B_1 B_2 \cdots B_k}$ can always be positive definite.

\end{document}